\documentclass[12pt]{article}
\usepackage{hyperref}
\usepackage{float}
\usepackage{amsmath}
\usepackage{graphicx,psfrag,epsf}
\usepackage{enumerate}
\usepackage{natbib}
\usepackage{multirow}
\usepackage[ruled,vlined]{algorithm2e}
\usepackage{array}
\usepackage{graphicx}
\usepackage[utf8 ]{inputenc}
\usepackage[english]{babel}
\usepackage{amsthm}
\usepackage[bottom]{footmisc}
\usepackage{longtable}
\usepackage{mathrsfs}
\usepackage{extarrows}
\usepackage{bigints}
\usepackage{subcaption}

\usepackage{hyperref}

\usepackage{setspace}

\usepackage[top=0.76in, bottom=0.76in, left=0.76in, right=0.76in]{geometry}
\usepackage{amsmath}
\usepackage{amsfonts}

\newtheorem{theorem}{Theorem}[section]
\newtheorem{lemma}[theorem]{Lemma}

\newtheorem{assumption}{Assumption}

\newtheorem{remark}{Remark}[section]

\newcommand\floor[1]{\lfloor#1\rfloor}

\begin{document}

\title{Principal Component Analysis using Frequency Components of Multivariate Time Series 
}

\author{
Raanju R. Sundararajan \\ Southern Methodist University }

\vspace{0.3cm}

\date{}

\maketitle

\begin{abstract}

\noindent Dimension reduction techniques for multivariate time series decompose the observed series into a few useful independent/orthogonal univariate components. We develop a spectral domain method for multivariate second-order stationary time series that linearly transforms the observed series into several  groups of lower-dimensional multivariate subseries. These multivariate subseries have non-zero spectral coherence among components within a group but have zero spectral coherence among components across groups. The observed series is expressed as a sum of frequency components whose variances are  proportional to the spectral matrices at the respective frequencies. The demixing matrix is then estimated using an eigendecomposition on the sum of the variance matrices of these frequency components and its asymptotic properties are derived. Finally, a consistent test on the cross-spectrum of pairs of components is used to find the desired segmentation into the lower-dimensional subseries. The numerical performance of the proposed method is illustrated through simulation examples and an application to modeling and forecasting wind data is presented.

\end{abstract}

\hrule
\hrulefill

\section{Introduction}

Analysis of multivariate time series is an interesting part of statistics with several important applications in areas such as engineering, economics, finance, neuroscience and geoscience. Time series data from these areas appear often in multivariate form thereby making the statistical analysis challenging from a computational and theoretical standpoint. As an example in environmental science, air pollution related variables such as $CO$ and $NO_2$ gathered from different monitoring stations across time result in a multivariate time series. The problem of importance here is to forecast these variables over the next few time steps. As an example in engineering, wind data recorded either hourly or sub-hourly from several geographical locations result in a multivariate time series. Modeling and forecasting this wind series is vital in successfully utilizing this energy source. To maintain a reliable energy network, it is necessary to plan for breaks in the influx of energy from this source and forecasting assists in managing the conventional and renewable generators in the network. An effective dimension reduction technique for the observed multivariate series would greatly assist in improving forecasting accuracy. Multivariate autoregressive moving average models (\citealp{brockwell_davis}), also known as VARMA, are commonly used to model dependence in multivariate stationary time series due to its straightforward implementation, interpretation and its ability to provide predictions. VARMA models, however, come with several challenges related to estimation, identifiability, lag order mis-specification and these challenges affect subsequent tasks such as forecasting. The multivariate nature of the time series data alone poses a computational difficulty for VARMA models. Dimension reduction hence becomes a very relevant and important problem in the analysis of multivariate time series.  

Among various dimension reduction methods, factors models and independent component analysis (ICA)  (\citealp{lam2012}, \citealp{mattesondoc}, \citealp{ombao_motta_2012})  focus on simplifying the analysis of multivariate time series  by linearly transforming the series into a few useful independent/orthogonal components or factors. Another approach is via the restriction of parameters in parametric multivariate time series models. \citet{davis_sparse_ar} propose an approach wherein the significant VAR (vector autoregressive) model coefficients are detected using a test on the partial spectral coherence among the components. Constrained maximum likelihood estimation of the VAR model coefficients is another technique proposed in Chapter 5 \citet{lutkepohl} and imposing structural restrictions on the VAR model is another approach in Chapter 9 of \citet{lutkepohl}.

Principal component analysis (PCA) for time series is another well known method for dimension reduction. The classical  dynamic PCA for time series from Chapter 9 in \citet{brillinger81} expresses the observed multivariate series as a two sided moving average of an uncorrelated vector process. This process is called the principle component series and has a diagonal spectral matrix with no spectral coherence between any two components. Chapter 9 of \citet{brillinger81} also considers a frequency-wise joint modeling of a frequency component of a $p$-variate series along with its Hilbert transform to obtain a $2p$-variate series. This $2p$-variate real valued series now has a real symmetric covariance matrix and a frequency-wise traditional PCA can be done on it. PCA in the multivariate time series, stationary and nonstationary, setting has been attempted in  \citet{stockwatson}, \citet{slex_ombao}, \citet{ombao_2006} wherein the objective is to find contemporaneous linear transforms of the observed series resulting in several orthogonal univariate subseries. Unlike the above mentioned works, \citet{yao2018} model the observed multivariate series as linearly generated (i.e using a mixing matrix) by a latent series that can be segmented into several multivariate subseries. Correlation exists between components within a subseries but not between components across subseries. Abbreviated as TS-PCA, their method performs an eigenanalysis on the sums of quadratic forms of lagged covariance matrices and the eigenvectors are utilized as an initial solution to the demixing matrix (inverse of the mixing matrix). Then tests of cross correlations are carried out to permute the rows of this demixing matrix and leading to the desired segmentation. As an extension to multivariate ICA,  \citet{cardoso1998} consider a similar setup but without considering the lagged dependence. \citet{beloucharani} pursue a related goal of diagonalization of lagged covariance matrices using an iterative optimization algorithm. Time domain methods listed here rely on pre-whitening (making the variance matrix diagonal) as a required first step. For example in \citet{yao2018}, the  component series are fitted with univariate autoregressive (AR) models with order chosen by AIC and the residuals from the fit are considered as the pre-whitened series. This task of fitting AR models to processes having strong periodocities or strong moving average components can be difficult as it can potentially lead large AR model orders.

We propose a new spectral domain method that models the observed multivariate second-order stationary time series as being linearly generated by a latent multivariate series. This latent series contains several subseries wherein components withing a subseries are allowed to have non zero spectral coherence but components from different subseries have zero spectral coherence. The key idea is to express the observed multivariate series as a sum of mutually exclusive and exhaustive frequency components. These frequency components are uncorrelated across unequal frequencies. First, an eigendecomposition on the sum of variance matrices of these frequency components results in an initial solution for the demixing  matrix and the recovered latent series. Second, a consistent test of spectral coherence on pairs of components of this recovered series is carried out leading to re-arrangement (permutation) of the components into the desired segmentation. The advantages of this approach are: a). Unlike the time domain methods above, our  spectral domain approach requires no pre-whitening as an initial step. This is particularly significant when dealing with stationary processes with strong periodocities wherein pre-whitening can be challenging because an attempt fit an AR model would potentially lead to very high orders. Our simulation results and discussions in Section \ref{sec:simulations} shed further light on this problem, b). With no pre-whitening required, our approach is completely nonparametric and unlike the conventional dynamic PCA for time series, no guarantees are put in place regarding the existence of a segmentation into lower-dimensional subseries and c). Our approach can further be modified to carry out PCA over specified frequency bands. This is particularly useful while analyzing certain types of data, such as those from neuroscience experiments or those for assessing wind turbine loads,  wherein there is great interest to understand band specific behavior.       

In Section \ref{sec:methodology} we start with the required preliminary notations and definitions and then  describe our proposed method in detail. The asymptotic properties of our method are discussed in Section \ref{sec:theory}. We compare the performance of our method with competitors using simulation examples in Section \ref{sec:simulations}. We discuss here the advantages of our spectral domain method over its time domain counterparts in models that have strong periodic components and/or strong moving average components. In Section \ref{sec:application} an application of our proposed method in forecasting wind data attributes is presented. We consider the multivariate modeling of sub-hourly time series data of wind speeds gathered across several wind farm locations in a given geographical region. Forecasting accuracy is used a measure of comparing the performance of the various methods. The concluding remarks are in Section \ref{sec:conclusion}.

\section{Methodology}
\label{sec:methodology}

In this section we describe our method to transform the observed multivariate time series into groups of lower-dimensional multivariate subseries that do not have any spectral coherence across the different subseries. The theoretical properties of our method is given in Section \ref{sec:theory}. We begin with the model setup followed by a description of the proposed technique.

 Let $X_t$, $t=1,2,\hdots,T$, the observed $p$-variate  zero mean second-order stationary time series  be written as 
\begin{equation}\label{eq:main_model}
X_t = A Y_t  = A \begin{bmatrix}
Y_{1,t} \\
Y_{2,t} \\
\vdots \\
Y_{m,t}
\end{bmatrix}
\end{equation}
where $Y_t$ is a $p$-variate zero mean second-order stationary time series containing the $m$ subseries $Y_{i,t} \in \mathbb{R}^{d_i}$, $i=1,2,\hdots,m$ and $A$ is the unknown invertible mixing matrix. The $m$ subseries $Y_{i,t}$ are such that there is no spectral coherence between components from different subseries. More precisely, for any two components $a^{(i)}$ and $b^{(j)}$, with $a^{(i)}$ belonging to subseries $Y_{i,t}$ and $b^{(j)}$ belonging to subseries $Y_{j,t}$ and $i \neq j$, we have  
\begin{equation} \label{eq:model_assumption_cross_spectral}
f_{Y,a^{(i)},b^{(j)}}(\omega) = 0 \;\; \textrm{and} \;\; R_{Y,a^{(i)},b^{(j)}} (\omega) = \frac{|f_{Y,a^{(i)},b^{(j)}}(\omega) |^2}{f_{Y,a^{(i)},a^{(i)}}(\omega) \; f_{Y,b^{(j)},b^{(j)}}(\omega)} =  0 \;\; \forall  \; \omega \in \lbrack -\pi, \pi \rbrack,
\end{equation}
where $f_Y(\omega) = \Big( f_{Y,a,b}(\omega) \Big)_{a,b=1,2,\hdots,p}$ is the $p \times p$ spectral matrix of $Y_t$ and $f_{Y,a,b}(\cdot)$ is the cross-spectrum between the series $Y_{a,t}$ and $Y_{b,t}$.    $R_{Y,a,b}(\omega)$ is the spectral coherence of the series $Y_{a,t}$ with series $Y_{b,t}$ at frequency $\omega$. The spectral coherence $R_{Y,a,b}(\omega)$ is an indicator of linear relationship between series $Y_{a,t}$ and $Y_{b,t}$. Further, by writing the autocovariance function $\gamma_{Y,a,b}(h) = cov(Y_{a,t},Y_{b,t+h})$  as $\gamma_{Y,a,b}(h) = \int_{-\pi}^{\pi} f_{Y,a,b}(\omega) e^{-ih\omega}d \omega$ we see that if two component series are uncorrelated, the cross-spectrum and spectral coherence are zero.

The $p \times p$ mixing matrix $A$ is assumed to be orthogonal with $A' A = I_p$. Under this setup, $A$ cannot be uniquely determined but $\mathcal{C}(A)$, the column space of $A$, can be uniquely determined. In view of the setup in \eqref{eq:main_model}, the matrix $A$ can be written as $A=(A_1,A_2,\hdots,A_m)$ where the matrix $A_i$, $i=1,2,\hdots,m$, has $p$ rows and $d_i$ columns. This results in
\begin{equation}\label{eq:latent_yt_representation}
Y_{i,t} = A_{i}'X_t \;\; \textrm{for}  \;\; i=1,2,\hdots,m.
\end{equation}

\noindent The objective here is to find $A$ and the resulting latent segmentation of $A'X_t = Y_t$ leading to the $m$ subseries $Y_{i,t}$, $i=1,2,\hdots,m$. Such a transformation provides a reduction in the dimension of the observed $p$-variate series into the $m$ subseries in lower dimensions thereby assisting in related tasks such as forecasting. It is important to note that no assumptions are made regarding existence of such a segmentation into the $m$ subseries i.e. $m=1$ is allowed as a possibility.

\noindent To achieve this segmentation we look into the spectral representation of a stationary time series (\citealp{brillinger81}, \citealp{brockwell_davis}). We have for $Y_t$ and $X_t$,

\begin{equation}\label{eq:spectral_representation}
Y_t = \int_{-\pi}^{\pi} e^{it\omega} dZ_Y(\omega) \;\; \textrm{and} \;\; X_t = \int_{-\pi}^{\pi} e^{it\omega} dZ_X(\omega),
\end{equation}    
where $Z_Y(\omega)$ and $Z_X(\omega)$ are orthogonal increment processes on $\lbrack -\pi, \pi  \rbrack$, $var \Big( dZ_Y(\omega) \Big) = f_Y(\omega)$ and $var \Big( dZ_X(\omega) \Big) = f_X(\omega)$ are the spectral density matrices. The frequency component of $Y_t$ corresponding to some frequency $\omega \in ( 0 ,\pi )$ can be written as
\begin{equation}\label{eq:frequency_components}
Y(t,\omega) = \int_{-\omega-\delta}^{-\omega+\delta} e^{it\lambda} dZ_Y(\lambda) +  \int_{\omega-\delta}^{\omega+\delta} e^{it\lambda} dZ_Y(\lambda)
\end{equation}
for some small $\delta>0$. Writing $U_Y(\omega) = Re\Big(Z_Y(\omega)\Big)$ and $V_Y(\omega) = -Im\Big(Z_Y(\omega)\Big)$  where $Re(\cdot)$ and $Im(\cdot)$ denote the real and imaginary parts, it can be seen that for a small $\delta > 0 $,
\begin{equation}\label{eq:freq_components_alt}
Y(t,\omega) \approx 2 \Big[ \cos(t\omega) dU_Y(\omega) + \sin(t\omega) dU_V(\omega)  \Big],
\end{equation} 
where $cov \Big( dU_Y(\omega), dU_Y(\lambda) \Big) = cov \Big( dV_Y(\omega) ,  dV_Y(\lambda) \Big) =  1_{ \{ \omega=\lambda \} } \frac{1}{2}   Re\Big( f_Y(\omega) \Big)$ and $cov \Big( dU_Y(\omega) , dV_Y(\lambda) \Big) = 1_{ \{ \omega = \lambda \} } \frac{1}{2} Im \Big( f_Y(\omega) \Big)$ for some $0< \omega,
\lambda < \pi$. Further, from \eqref{eq:spectral_representation}-\eqref{eq:freq_components_alt} it can be observed that the series $Y_t$ can be written as a sum of mutually exclusive and exhaustive frequency components 
\begin{equation}\label{eq:yt_sum_of_frequency_components}
Y_t = \sum_{j=1}^F Y(t,\omega_j)
\end{equation}
where $F = \frac{\pi}{2\delta}$ and $\{ \omega_1, \omega_2, \hdots, \omega_F \}$ is a discretized set of  frequencies in $(0,\pi)$. Here,  for $j \neq k$ and any $t,s$ we have  
\begin{equation}
cov \Big( Y(t,\omega_j), Y(s,\omega_k) \Big) = 0.
\end{equation}
See Section 4.6 of \citet{brillinger81}. From \eqref{eq:freq_components_alt} it follows that the variance matrix of the frequency component $Y(t,\omega)$ is proportional to $Re \Big( f_Y(\omega) \Big)$. 

With $Y_t$ assumed to have the latent segmentation structure given in \eqref{eq:main_model}, $var \Big( Y(t,\omega) \Big) = 4\delta Re\Big( f_Y(\omega) \Big)$ can be treated as a real symmetric block-diagonal nonnegative definite matrix. In view of \eqref{eq:main_model}, the frequency component of $X_t$, namely $X(t,\omega)$, and its variance can be written as 
\begin{equation}\label{eq:xt_freq_component}
X(t,\omega) = A \; Y(t,\omega)\;\; \textrm{and} \;\;  Re \Big( f_X(\omega) \Big) = A \; Re \Big( f_Y(\omega)  \Big) A'. 
\end{equation}
 
\noindent To uncover the latent segmentation in $Y_t$, we consider the $p \times p$ matrices $S_X$ and $S_Y$ given by 
\begin{equation}\label{eq:sx_sy_matrices}
S_X = \sum_{j=1}^F Re \Big( f_X(\omega_j) \Big) = A \; S_Y A' = A  \Big[ \sum_{j=1}^F Re \Big( f_Y(\omega_j) \Big) \Big] A'.
\end{equation} 
where $Re(\cdot)$ denotes the real part. Observe that $S_X$ and $S_Y$ are sums of $F$ real symmetric nonnegative definite matrices. The eigendecompositions $S_X = L_X D_X L_X' $ and $S_Y = L_Y D_Y L_Y'$, where $L_X$ and $L_Y$ are the orthogonal matrices of eigenvectors of $S_X$ and $S_Y$ respectively and $D_X$ and $D_Y$ are diagonal matrices with diagonal entries as eigenvalues of $S_X$ and $S_Y$ respectively, lead to 
\begin{equation}\label{eq:lx_ly_relationship}
L_X' X_t = L_Y' Y_t.
\end{equation} 

With $Y_t$ having the latent segmentation structure given in \eqref{eq:main_model}, $S_Y$ can be partitioned into blocks of positive definite matrices i.e $S_Y = \mbox{diag}(S_{Y,1},S_{Y,2},\hdots,S_{Y,m})$ wherein the eigenvectors and eigenvalues of the blocks $S_{Y,i}$, $i=1,2,\hdots,m$, are that of $S_Y$ as well. Assuming two blocks do not share an eigenvalue, Proposition 1 in \citet{yao2018}  implies that $L_Y$ is a block-diagonal orthogonal matrix with the same block structure as $S_Y$. However, this ordering of the blocks in $S_Y$ is unknown and a column permutation of $L_Y$ leads to a block-diagonal orthogonal matrix. Thus \eqref{eq:lx_ly_relationship} implies that re-arranging the components of $L_X'X_t$ results in the desired segmentation into the $m$ groups (subseries) of components in \eqref{eq:main_model}. In Section \ref{sec:permuation_subseries}, we describe our method to find these $m$ groups from the components of $L_X' X_t$ such that that there is no spectral coherence between components from different groups (subseries). 

Note that the above approach leading to \eqref{eq:lx_ly_relationship} is similar to that of \citet{yao2018} except that the sum of quadratic forms of lagged covariance matrices in their work is replaced by the sum of the real parts of spectral matrices at various frequencies. Unlike their work, our subsequent testing procedure described in Section \ref{sec:permuation_subseries} that  finds groups among components of $L_X^{'} X_t$ does not require pre-whitening of the observed series. The finite sample performance of their approach requires pre-whitening to ensure meaningful comparisons between the different pairs in $L_X^{'} X_t$. In order to find the permutation (re-arrangement of the components of $L_X^{'} X_t$ leading to the desired segmentation), their method fits AR models to individual components using AIC criterion and utilizes the residuals. Certain ARMA processes that exhibit strong pseudo-periodicities and strong moving average components pose a great challenge to pre-whiten and often result  in a very large AR model order. Our approach is completely nonparametric as it avoids any kind of pre-whitening and is more suitable for processes that contain strong periodicities. In Section \ref{sec:simulations}, Models 2 and 3  have strong periodic components and we illustrate the better performance of our method in these situations.

\begin{remark}
\noindent (a). The $S_X$ and $S_Y$ matrices are sums of variances of the frequency components. Similar to the argument on the lag selection of cross-covariance matrices in \citet{yao2018}, we note here that the choice for $F$ is not a sensitive choice because we only need an adequate set of frequencies carrying information on the block structure in $S_Y$. In small sample situations in Sections \ref{sec:simulations}, \ref{sec:application}, we consider a discretized sequence of the interval $(0, \pi)$ as a proxy for the $F$ frequencies $ \{ \omega_1, \omega_2, \hdots, \omega_F \}$ needed in \eqref{eq:sx_sy_matrices}. Our simulation studies with different choices of the discretized interval show the lack of sensitivity of our method towards this choice.

\noindent (b). In our formulation leading to \eqref{eq:sx_sy_matrices}, we only consider the sum of variance matrices of the frequency components. It can be seen from \eqref{eq:frequency_components}-\eqref{eq:freq_components_alt} that $cov(Y(t,\omega), Y(t+h , \omega))$ for $h > 0$ involves the term
\begin{equation}
\int_{\omega - \delta}^{\omega+ \delta} \cos(h \lambda) Re \Big( f_Y(\lambda) \Big) d\lambda.
\end{equation}
Since the expression for these lagged covariance matrices involves again the real part of the spectral matrices, it is viewed as surplus to the useful information needed to uncover the block structure in the spectral matrix.  
   
\end{remark}

\noindent In order to estimate the quantities in $S_X$ and $L_X$ we first obtain a kernel estimator of the spectral matrix. We first define the discrete Fourier transform (DFT) and the periodogram of a p-variate series as
\begin{equation*}
 J_X(\omega)= \frac{1}{\sqrt{2\pi T}} \sum_{t=1}^{T}X_t e^{-it\omega},\quad I_{X}(\omega)=J_X(\omega)J_X(\omega)^*,
\end{equation*}
\noindent where $J_X(\omega)^*$ denotes the conjugate transpose, and the estimated $p \times p$ spectral density matrix, for $\omega \in \lbrack -\pi,\pi \rbrack$ is given by
\begin{equation} \label{eq:kernel_spectral_estimator}
\widehat{f}_X(\omega)= \frac{1}{T} \sum_{j=-\floor{\frac{T-1}{2}}}^{\floor{\frac{T}{2}}} \;  \; K_h( \omega - \omega_j) \;  I_{X}(\omega_j),
\end{equation}
where  $\omega_j=\frac{2 \pi}{T} j$ and $K_h(\cdot)=\frac{1}{h}K(\frac{\cdot}{h})$ where $K(\cdot)$ is
a nonnegative symmetric kernel function and $h$ denotes the bandwidth. Certain assumptions on the kernel and bandwidth are needed to establish large sample results and this is discussed in Section \ref{sec:theory}. A sample version of $S_X$ and $L_X$ are then given by   

\begin{equation}\label{eq:sx_matrix_estimated}
\widehat{S}_X = \sum_{j=1}^F Re \Big( \widehat{f}_X(\omega_j) \Big)  = \widehat{L}_X \widehat{D}_X \widehat{L}_X' 
\end{equation}   
where the $p \times p$ matrices $\widehat{L}_X$ and $\widehat{D}_X$ are the matrix of eigenvectors and  diagonal matrix of eigenvalues of $\widehat{S}_X$, respectively. We then consider the components of $\widehat{L}_X^{'} X_t$ and find a permutation of these components leading to the required $m$ subseries wherein there is no spectral coherence between components from different groups. This procedure along with the estimation of $m$ is described next in Section \ref{sec:permuation_subseries}.

\subsection{Finding the $m$ subseries using pairwise testing}   
\label{sec:permuation_subseries}

In this section we describe the method to permute the $p$ components of $L_X' X_t$ to obtain the $m$ subseries $Y_{i,t} \in \mathbb{R}^{d_i}$, $i=1,2,\hdots,m$. We begin with describing the test of spectral coherence between two component series. Unlike the time domain cross-correlation test suggested in \citet{yao2018}, pre-whitening of the components is not a necessary first step for this test.

\noindent Denote $\tilde{Y}_t = L_X' X_t$ and let $f_{\tilde{Y}}(\cdot)$ be its $p \times p$ spectral matrix. For any two univariate component series $\tilde{Y}_{a,t}$ and $\tilde{Y}_{b,t}$, $a,b=1,2,\hdots,p$ and $a \neq b$, we wish to test $H_0:  R_{ \tilde{Y}, a,b} (\omega) = 0,\; \forall \omega \in \lbrack -\pi , \pi\rbrack$ i.e the spectral coherence being zero at all frequencies. As the test statistic, we use the metric
\begin{equation} \label{eq:cross_spectrum_test_statistic}
D(\tilde{Y},a,b) = \int_{-\pi}^{ \pi}  \frac{ |  f_{\tilde{Y},a,b}(\omega)|^{2} }{f_{\tilde{Y},a,a}(\omega)f_{\tilde{Y},b,b}(\omega)}  \;d\omega,
\end{equation}
\noindent where $f_{\tilde{Y},a,b}$ denotes entry $(a,b)$ in the spectral matrix of $\tilde{Y}_t$. An estimated version of the above quantity, denoted as $\widehat{D}(\tilde{Y},a,b)$, can be obtained by plugging in the kernel spectral estimator of $f_{\tilde{Y}}(\omega)$ defined in \eqref{eq:kernel_spectral_estimator}. The large sample distribution of the test statistic in \eqref{eq:cross_spectrum_test_statistic} is discussed in Section \ref{sec:theory} and this result yields the critical values of the test needed in Sections \ref{sec:simulations}, \ref{sec:application}.

After determining the statistical significance of every pair of components in $\tilde{Y}_t = L_X' X_t$, we are left with a graph $G=(V,E)$ wherein the vertex set $V$ corresponds to the set of components in $\tilde{Y}_t$ and the $p \times p$ adjacency matrix $E = (e_{a,b})$ is such that $e_{a,b} = 1$ when the test based on $D(\tilde{Y},a,b)$ resulted in a rejection of $H_0$. Finally, in order to obtain the $m$ subseries of components, we find connected components in $\tilde{Y}_t$. More precisely, two components $a$ and $b$ are placed in the same subseries if either 
\begin{itemize}
\item[(a).] $e_{a,b}=1$ or
\item[(b).] There exists $\{ v_1,v_2,\hdots v_Q \} \subset \{1, 2, \hdots, p \}$ such that $e_{a,v_1} = e_{v_1,v_2} = \hdots = e_{v_{Q-1},v_Q} = e_{v_Q, b} = 1$ for some $Q<p-1$.    

\end{itemize}

\noindent An estimate of this graph and its adjacency matrix lead us to an estimate of $m$, the number of subseries in \eqref{eq:main_model}. The asymptotic results concerning this estimation is provided in Section \ref{sec:theory}. We now include Algorithm \ref{algo:two_step_procedure} that summarizes the steps leading to an estimate of the demixing matrix $A^{-1}$ and the number of subseries $m$.

\begin{algorithm}[H] \label{algo:two_step_procedure}
\caption{Estimating the demixing matrix $A^{-1}$ and number of subseries $m$ from \eqref{eq:main_model}.}
\SetAlgoLined
 \textbf{Output}: Estimates $\widehat{A}^{-1}$ and $\widehat{m}$.  \\
 \textbf{Input}: $p$-variate time series data $X_t$, $t=1,2,\hdots,T$, the discretized set of frequencies in $(0,\pi)$ given by $\mathcal{W} = \{\omega_1 < \omega_2 < \hdots \leq  \omega_{F} \}$, the kernel $K(\cdot)$ and bandwidth $h$.  
 \vspace{0.3cm}
 \begin{itemize}
 \item[1:] Calculate the estimates $\widehat{f}_X(\omega)$ in \eqref{eq:kernel_spectral_estimator} for every $\omega \in \mathcal{W}$ and obtain $\widehat{S}_X$ in \eqref{eq:sx_matrix_estimated}. Compute the Eigendecomposition $\widehat{S}_X = \widehat{L}_X \widehat{D}_X \widehat{L}_X^{'}$. Find the series $\tilde{Y}_t = \widehat{L}_X^{'} X_t$ for $t=1,2,\hdots,T$.   
 \item[2:] For any two univariate component series $\tilde{Y}_{a,t}$ and $\tilde{Y}_{b,t}$, $a,b=1,2,\hdots,p$ and $a \neq b$, test for zero spectral coherence ($H_0$) using the statistic in \eqref{eq:cross_spectrum_test_statistic}. 
 \item[3:] Create graph $G=(V,E)$, the vertex set $V$ corresponds to the set of components in $\tilde{Y}_t$, the adjacency matrix $E = (e_{a,b})$ is such that $e_{a,b} = 1$ when the test from Step 2 resulted in a rejection of $H_0$ (zero spectral coherence).  
 \item[4:] Using $G$ from Step 3, two components $a$ and $b$ are placed in the same subseries if either $e_{a,b}=1$ or there exists $\{ v_1,v_2,\hdots v_Q \} \subset \{1, 2, \hdots, p \}$ such that $e_{a,v_1} = e_{v_1,v_2} = \hdots$ $= e_{v_{Q-1},v_Q} = e_{v_Q, b} = 1$ for some $Q<p-1$. This leads to an estimate $\widehat{m}$ of $m$. 
\item[5:] Using grouping of components in Step 4 find permutation of components of $\tilde{Y}_t$ as $\{1,2,\hdots,p \}$ $\rightarrow \{ \pi(1), \pi(2),\hdots,\pi(p) \} $, obtain the corresponding permutation matrix $P_{\pi}$. Output the estimated demixing matrix $\widehat{A}^{-1} = P_{\pi}^{'} \widehat{L}_X^{'}$.  
  \end{itemize}
\end{algorithm}

\begin{remark}[\bf Frequency band specific PCA]
Certain applications demand the use of a frequency band specific analysis of the multivariate time series.  Neuroscience experiments resulting in data such as the EEG or local field potentials (LFP) contain important information over various known frequency bands like Theta, Alpha, Beta and Gamma. Our approach can be extended to uncover the $m$ latent subseries $Y_t = (Y_{1,t}^{'} , Y_{2,t}^{'} , \hdots, Y_{m,t}^{'})^{'}$ that exist only at a specified frequency band.  Similar to \eqref{eq:model_assumption_cross_spectral}, for any two components $a^{(i)}$ and $b^{(j)}$, with $a^{(i)}$ belonging to subseries $Y_{i,t}$ and $b^{(j)}$ belonging to subseries $Y_{j,t}$ and $i \neq j$, we have  
\begin{equation} \label{eq:model_assumption_cross_spectral_band_specific}
 R_{Y,a^{(i)},b^{(j)}} (\omega) = \frac{|f_{Y,a^{(i)},b^{(j)}}(\omega) |^2}{f_{Y,a^{(i)},a^{(i)}}(\omega) \; f_{Y,b^{(j)},b^{(j)}}(\omega)} =  0 \;\; \forall  \; \omega \in ( \omega_1 , \omega_2 ),
\end{equation}
for some $0 < \omega_1 < \omega_2 < \pi$. To uncover the latent segmentation here we can consider the  matrix $S_X^{(\omega_1,\omega_2)} = \sum_{j=1}^{B} Re \Big( f_X(\lambda_j) \Big)$ where $ \{ \lambda_1, \lambda_2, \hdots, \lambda_B \}$ is a discretized set of frequencies from the interval $(\omega_1,\omega_2)$ and its eigendecomposition $S_X^{(\omega_1,\omega_2)} = L_X^{(\omega_1,\omega_2)} D_X^{(\omega_1,\omega_2)} ( L_X^{(\omega_1,\omega_2)} )^{'}$. Then with the series $\tilde{Y}_t = ( L_X^{(\omega_1,\omega_2)} )^{'} X_t$, for any two univariate component series $\tilde{Y}_{a,t}$ and $\tilde{Y}_{b,t}$, $a,b=1,2,\hdots,p$ and $a \neq b$, we wish to test $H_0: R_{ \tilde{Y}, a,b} (\omega) = 0, \; \forall \omega \in (\omega_1,\omega_2)$. The test statistic from \eqref{eq:cross_spectrum_test_statistic} is now given by 
\begin{equation} \label{eq:cross_spectrum_test_statistic_band_specific}
D_{\omega_1,\omega_2}(\tilde{Y},a,b) = \int_{\omega_1}^{ \omega_2}  \frac{ |  f_{\tilde{Y},a,b}(\omega)|^{2} }{f_{\tilde{Y},a,a}(\omega)f_{\tilde{Y},b,b}(\omega)}  \;d\omega.
\end{equation}
The asymptotic properties of the statistic above in \eqref{eq:cross_spectrum_test_statistic_band_specific}  follow similar to the results of the test statistic in \eqref{eq:cross_spectrum_test_statistic} given in Section \ref{sec:theory}.   
\end{remark}

\subsection{Theoretical properties}
\label{sec:theory}
Here we discuss the theoretical results associated with our proposed method. We begin with the properties of the estimator of the column space $\mathcal{C}(A)$ of the mixing matrix $A$. 

Let $\widehat{A} = (\widehat{A}_1, \widehat{A}_2,\hdots, \widehat{A}_m)$ be the estimate of the mixing matrix $A$. For any two half-orthogonal matrices $B_1, B_2 \; \in \mathbb{R}^{p \times r}$ satisfying $B_1^{'} B_1 = B_2^{'} B_2 = I_r$, we take the distance measure between their column spaces as
\begin{equation}\label{eq:subspace_distance_metric}
M \Big( \mathcal{C}(B_1) , \mathcal{C}(B_2)  \Big) = \sqrt{1 - \frac{1}{r}tr \Big( B_1 B_1^{'} B_2 B_2^{'} \Big) }
\end{equation}
where $tr(\cdot)$ denotes the trace of a matrix. $M \Big( \mathcal{C}(B_1) , \mathcal{C}(B_2)  \Big) = 0$ if the two column spaces are the same and is equal to 1 if they are orthogonal; see \citet{pan_yao_08}. The aim is to show $\mathcal{C}(\widehat{A}_i)$ is consistent for $\mathcal{C}(A_i)$ for each $i=1,2,\hdots,m$. We take $\widehat{L}_X$ from \eqref{eq:sx_matrix_estimated}, upto a column permutation, as a viable estimator for $A$ and show that its columns space consistently estimates $\mathcal{C}(A)$. 

\begin{assumption}\label{as:eigengap}
The $m$ block matrices in $S_Y = \mbox{diag}(S_{Y,1},S_{Y,2},\hdots,S_{Y,m})$ do not share any common eigenvalues and the distance between their smallest eigenvalues is positive. More precisely,
\begin{equation}
\min_{i \neq j}  \min_{\mu \in l( S_{Y,i} ), \; \nu \in l( S_{Y,j} ) } |\mu - \nu| \; > \; 0 
\end{equation}
where $l( S_{Y,i} )$ and $l( S_{Y,j} )$ denote the set of eigenvalues of the matrices $S_{Y,i}$ and $S_{Y,j}$ respectively. 
\end{assumption}

The above assumption relates to the sensitivity to perturbation of an invariant subspace dependent on eigenvalue separation (eigengap); see Chapter 8.1.3 of \citet{golub2012}. We discuss the eigenvalue separation from a computational standpoint in Section \ref{sec:simulations}. There, various choices for the vector process $Y_t$ in model \eqref{eq:main_model} is chosen leading to different eigenvalue separations (eigengap) and we evaluate  the performance of our method in those situations. We will also provide recommendations to select the bandwidth $h$ of the kernel estimator in \eqref{eq:kernel_spectral_estimator} based on the eigenvalue separation. 

 The next two standard assumptions concern the covariance structure of $X_t$ and the kernel and bandwidth used in the spectral matrix estimator in \eqref{eq:kernel_spectral_estimator}. 

\begin{assumption}\label{as:xt}
Let  $X_t$  be a $p$-variate zero mean second-order stationary time series with covariance function $\gamma(h) = cov(X_t,X_{t+h})$ satisfying $\sum_{h=-\infty}^{\infty} |h| \; |\gamma(h)| < \infty$. Further, its spectral matrix $f_X(\omega)$ is non-singular at all $\omega \in \lbrack -\pi ,\pi \rbrack$. 
\end{assumption}

\begin{assumption}\label{as:kernel}
(a). The kernel function $K(\cdot)$ used in \eqref{eq:kernel_spectral_estimator} is bounded, symmetric, nonnegative and Lipschitz-continuous with compact support $\lbrack -\pi ,  \pi \rbrack$ and
\vspace{-0.2cm}
\begin{equation*}
\int_{-\pi}^{ \pi} K(\omega) \textrm{d}\omega = 1.
\vspace{-0.2cm}
\end{equation*}
where $K(\omega)$ has a continuous Fourier transform  $k(u)$ such that
\vspace{-0.2cm}
\begin{equation*}
\int k^2(u)\textrm{d}u < \infty \;\; \textrm{and} \;\; \int k^4(u)\textrm{d}u < \infty.
\vspace{-0.2cm}
\end{equation*}
(b). The bandwidth $h$ is such that $h^{9/2}T \rightarrow 0$ and $h^2 T \rightarrow \infty$ as $T \rightarrow \infty$.
\end{assumption}

\begin{theorem}\label{thm:colspace_consistency} Suppose that Assumptions \ref{as:eigengap}, \ref{as:xt}, \ref{as:kernel} are satisfied and let $\widehat{A} = (\widehat{A}_1, \widehat{A}_2,\hdots, \widehat{A}_m)$ be the estimated mixing matrix after permuting the columns of $\widehat{L}_X$ in \eqref{eq:sx_matrix_estimated}. Then 
\begin{equation}
\max_{i=1,2,\hdots,m} M \Big( \mathcal{C} (A_i) , \mathcal{C} (\widehat{A}_i)  \Big) = O_p \Big( \frac{1}{ \sqrt{Th} } \Big).
\end{equation} 
\end{theorem}

\begin{proof}
See Appendix for details of the proof.
\end{proof}

Next, we state the large sample result of the test statistic used in the pairwise test described in Section \ref{sec:permuation_subseries}. In addition to Assumptions \ref{as:xt}, \ref{as:kernel}, let $X_t$ have finite moments of all orders. Then under $H_0:  R_{ \tilde{Y}, a,b} (\omega) = 0,\; \forall \omega \in \lbrack -\pi , \pi\rbrack$, by an application of Theorem 2.5 of \citet{eichler08}, an estimated version of the test statistic in \eqref{eq:cross_spectrum_test_statistic} is such that
\begin{equation}\label{eq:test_statistic_null_distribution}
T \sqrt{h} \widehat{D}(\tilde{Y},a,b) - \frac{\mu_0}{\sqrt{h}} \overset{d}{\longrightarrow} N(0 , \sigma_0^2)
\end{equation} 
where $\mu_0 = \int_{-\pi}^{\pi} K^2(u) du  $ and $\sigma_0^2 = 2 \pi \int_{-2 \pi}^{2 \pi}\; \Big( \int_{- \pi}^{  \pi} K(u)K(u+v) du \Big)^2\; dv $ and $\overset{d}{\longrightarrow}$ denotes convergence in distribution. For the simulation study and applications in Sections \ref{sec:simulations}, \ref{sec:application}, we utilize the above result to obtain critical values for the pairwise tests. Next, under the alternative hypothesis $H_1:  R_{ \tilde{Y}, a,b} (\omega) \neq 0,\; \textrm{for some} \; \omega \in \lbrack -\pi , \pi\rbrack$, application of Theorem 5.1 in \citet{eichler08} shows that
\begin{equation} \label{eq:test_statistic_alt_result}
T h \widehat{D}(\tilde{Y},a,b) \overset{p}{\longrightarrow}  +\infty 
\end{equation}
thereby resulting in a consistent test. 

Finally, we state the result concerning the estimation of the true adjacency matrix $E$ described in Section \ref{sec:permuation_subseries}. Recall that our method conducts $p_0 = p(p-1)/2$ pairwise tests on the components of $\tilde{Y}_t = L_X' X_t$ with the test statistic $\tilde{D}$ in \eqref{eq:cross_spectrum_test_statistic}. Denote $\alpha_{a,b}$ as the significance level chosen while testing for zero spectral coherence between components $a$ and $b$. Let $\widehat{E}$ be the estimated adjacency matrix obtained by carrying out pairwise testing on the components of $\widehat{L}_X^{'}X_t$. 

\begin{theorem} \label{thm:adjacency_consistency}
Suppose that Assumptions \ref{as:eigengap}, \ref{as:xt}, \ref{as:kernel}  are satisfied. Then, as $T \rightarrow \infty
$ we have,
\begin{itemize}

\item[(a).] $P\Big( \widehat{E} = E \Big) \rightarrow 1$  if  $m=1$.
\item[(b).] $P\Big( \widehat{E} = E \Big) \geq 1- \sum\limits_{ \substack{a < b } } \alpha_{a,b}$  if  $m=2,3,\hdots,p$.
\end{itemize}
where $m$ denotes the true number of subseries among components of $Y_t$ defined in \eqref{eq:main_model}. 
\end{theorem}
\vspace{-0.2cm}
\begin{proof}
See Appendix for details of the proof. 
\end{proof}

\noindent The proof stems from the asymptotic results of the test statistic under the null and alternative  given in \eqref{eq:test_statistic_null_distribution} and \eqref{eq:test_statistic_alt_result}. When the true number of underlying subseries $m$ is more than one, the presence of multiple tests leaves the above result with a lower bound on the probability of correct adjacency matrix detection. A correction to the choice of the significance levels or a false discovery rate (FDR) based approach can potentially be adopted and requires future investigation. Also with the possibility of $m=1$, observe that unlike traditional PCA, no assumption is made regarding the existence of a lower-dimensional subseries that has no correlation with other subseries.

\section{Simulation Study}
\label{sec:simulations}

In this section we illustrate the performance of our method using a few simulation examples. Following  \citet{yao2018}, an estimate of the mixing matrix $A$, after permuting the columns using the pairwise testing procedure described in \ref{sec:permuation_subseries}, given by $\widehat{A} = (\widehat{A}_1,\widehat{A}_2,\hdots,\widehat{A}_{\widehat{m}}) $ is regarded as a \underline{\textit{correct segmentation}} if i). $\widehat{m} = m$, ii). rank$(\widehat{A}_j)$ = rank$(A_j)$ for each $j=1,2,\hdots,m$ and iii). $M^2( \mathcal{C}(\widehat{A}_j) , \mathcal{C}(A_j)  ) = \min_{1 \leq i \leq m} M^2( \mathcal{C}(\widehat{A}_i) , \mathcal{C}(A_i)  )$ where the metric $M$ is defined in \eqref{eq:subspace_distance_metric}. We simulate from the below defined models 200 times for series lengths $T=200,500,1000$ and report the percentage of \textit{correct segmentation}. In addition, the average and maximum estimation errors given by 
\begin{equation} \label{eq:sim_error_definitions}
\frac{1}{m}\sum_{j=1}^m M^2( \mathcal{C}(\widehat{A}_j) , \mathcal{C}(A_j)  ) \;\; \textrm{and} \;\;  \max_{1 \leq j \leq m} M^2( \mathcal{C}(\widehat{A}_j) , \mathcal{C}(A_j)  )
\end{equation}
are also reported. Note that summaries from these error measures are reported only for replications leading to a \textit{correct segmentation}. 

\noindent We consider the model in \eqref{eq:main_model} with the $p \times p$ orthogonal matrix  $A$ being randomly generated using the technique from \citet{rortho}. We select the Bartlett-Priestley kernel (\citealp{priestley84}) for estimating the spectral matrix in \eqref{eq:kernel_spectral_estimator}. This kernel $K(\cdot)$ and its Fourier transform $k(\cdot)$ are given by 
\begin{equation}\label{eq:bp_kern_fourier_details}
K(\theta) = \frac{3}{4 \pi} \Big( 1-\frac{\theta^2}{\pi^2} \Big) 1_{\lbrack -\pi,\pi \rbrack }(\theta), \;\; k(u) = \frac{3}{\pi^2 u^2} \Big( \frac{\sin(\pi u)}{\pi u} - \cos (\pi u) \Big).
\end{equation}
It can be verified from the above that this kernel choice and its Fourier transform satisfy Assumption \ref{as:kernel}(a). For carrying out the $p(p-1)/2$ pairwise tests of zero spectral coherence (Step 2 of Algorithm \ref{algo:two_step_procedure}), we obtain false discovery rate (FDR) driven corrections to the p-values from the tests (\citealp{fdr_by}). We adopt this correction and under independence of the test statistics or a positive regression dependence, this correction keeps the error rate below the prescribed level of 0.05.

\subsection{Bandwidth selection} \label{sec:bandwidth_slection}
One important tuning parameter that needs selection is the bandwidth $h$ used in the kernel estimator in \eqref{eq:kernel_spectral_estimator}. The size and power of the test described in Section \ref{sec:permuation_subseries} depends on $h$. Additionally, and equally important, the rate given in Theorem \ref{thm:colspace_consistency} along with the eigenvalue separation in Assumption \ref{as:eigengap} in finite sample situations also depends on $h$. Bandwidth selection methods for nonparametric kernel spectral estimators exist in the literature. \citet{beltrao_bandwidth} and \citet{buhlmann_bandwidth} are a couple of examples wherein the former propose a cross-validated approach while the latter carry out an iterative method to selecting locally optimal window widths. Our approach here aims to find a compromise between the `conflicting' requirements on the bandwidth $h$. More precisely, in Assumption \ref{as:kernel} the bandwidth $h$ tends to zero at a certain rate. This is needed for the large sample result in \eqref{eq:test_statistic_null_distribution}. In contrast, the eigengap requirement from Assumption \ref{as:eigengap} and the consistency result in Theorem \ref{thm:colspace_consistency} indicate that excessively small bandwidths are undesirable as it affects the consistency of the estimated column space $\mathcal{C}(\widehat{A}_i)$ for each $i=1,2,\hdots,m$.

We first report the eigenvalue separation (eigengap) of all the models for an estimate of the matrix $S_Y$ in \eqref{eq:sx_sy_matrices} using a kernel spectral estimator at different bandwidths. The sequence of bandwidths considered here are such that they satisfy Assumption \ref{as:kernel} in Section \ref{sec:theory}. We observe from Figure \ref{fig:eigengap_models} that the eigengap is higher at larger  bandwidths. In our simulation study, we considered a sequence of bandwidths $h = T^{-q}$ where $q$ is obtained from a sequence $\{ 0.1,0.15,\hdots,0.45 \}$ and we observed that our method performs best for larger bandwidths (smaller values of $q$)  and this is in agreement with the result in Theorem \ref{thm:colspace_consistency}. However an excessively large bandwidth affects the asymptotic result of the test statistic in \eqref{eq:test_statistic_null_distribution}. Hence in finite sample situations, we recommend a choice of bandwidth $h=T^{-q}$ where $q \approx 0.15$.

Another kernel that satisfies Assumption \ref{as:kernel} is the Parzen kernel (see pg. 448 of \citealp{priestley84}). In finite sample situations, this kernel choice produced accuracy results similar to the Bartlett-Priestley kernel. The recommended bandwidth with this kernel choice is again $h=T^{-q}$ where $q \approx 0.15$

\begin{figure}[H]
\centering
\includegraphics[width = 3.2in , height = 2in]{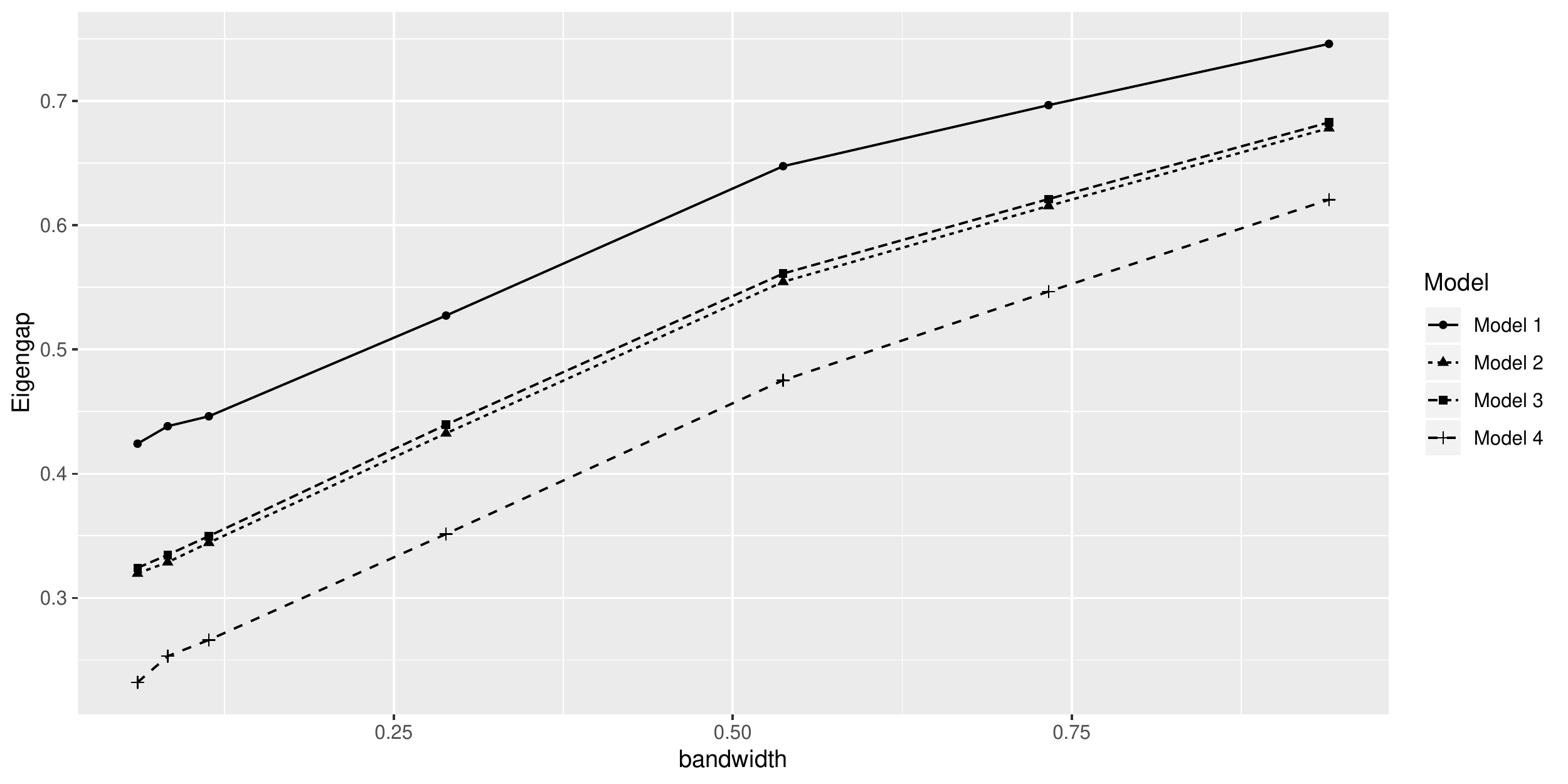}
\caption{Eigenvalue separation (eigengap) among the blocks of $S_Y$ obtained using a kernel spectral estimator of the spectral matrix at various bandwidths (x-axis). Eigengaps are based on averages over 100 replications of the various models with series length $T=500$.  } \label{fig:eigengap_models}
\end{figure}    


Next, we simulate from 5 models at various series lengths and present i) the average and maximum errors from  the measure $M^2$ defined in \eqref{eq:sim_error_definitions} and ii) the percentage of \textit{correct segmentation} out of 200 replications. Models 1 and 2 are of dimension $p=6$ with the former being the same as Example 1 of \citet{yao2018} and Models 3 and 4 are of dimension $p=9$. Models 2 and 3 are chosen such that the former has a strong moving average component and the latter has a component with strong pseudo-periodocity. In these two models we witness a better performance of the frequency domain method FC-PCA in contrast to the time domain method. Model 5 is of dimension $p=7$ and is the case where $m=1$ i.e the segmentation into the lower dimensional subseries does not exist.

\noindent \textbf{Model 1}: We take $p=6$, $m=3$ and the components of $Y_t$ are given by $Y_{k,t} = z_{1,t+k-1}$ for $k=1,2,3$, $Y_{k,t} = z_{2,t+k-4}$ for $k=4,5$ and $Y_{k,t} = z_{3,t}$ for $k=6$. Here, $z_{1,t}$ follows a ARMA$(2,4)$ with AR coefficients $(0.5,0.3)$ and MA coefficients $(-0.9,0.3,1.2,1.3)$ and innovations following i.i.d $N(0,1)$, $z_{2,t}$ follows a ARMA$(2,3)$ with AR coefficients $(0.8,-0.5)$ and MA coefficients $(1,0.8,1.8)$ and innovations following i.i.d $N(0,3)$, $z_{3,t}$ follows a ARMA$(2,2)$ with AR coefficients $(-0.7,-0.5)$ and MA coefficients $(-1,-0.8)$ and innovations following i.i.d $N(0,5)$. 

\begin{figure}[H]
\begin{minipage}{0.45\textwidth}
\begin{tabular}{|c|c|c|c|}
\hline
T & Method & max. $M^2$ & avg. $M^2$ \\
\hline
\multirow{2}{*}{200} & FC-PCA & 0.225 & 0.139 \\
 & TS-PCA & 0.142 & 0.085  \\
\hline
\multirow{2}{*}{500} & FC-PCA & 0.061 & 0.040  \\
 & TS-PCA & 0.076 & 0.044  \\
\hline
\multirow{2}{*}{1000} & FC-PCA & 0.028 & 0.018  \\
 & TS-PCA & 0.046 & 0.025  \\
\hline
\end{tabular}
\end{minipage}
\begin{minipage}{0.5\textwidth}
\includegraphics[scale=0.4]{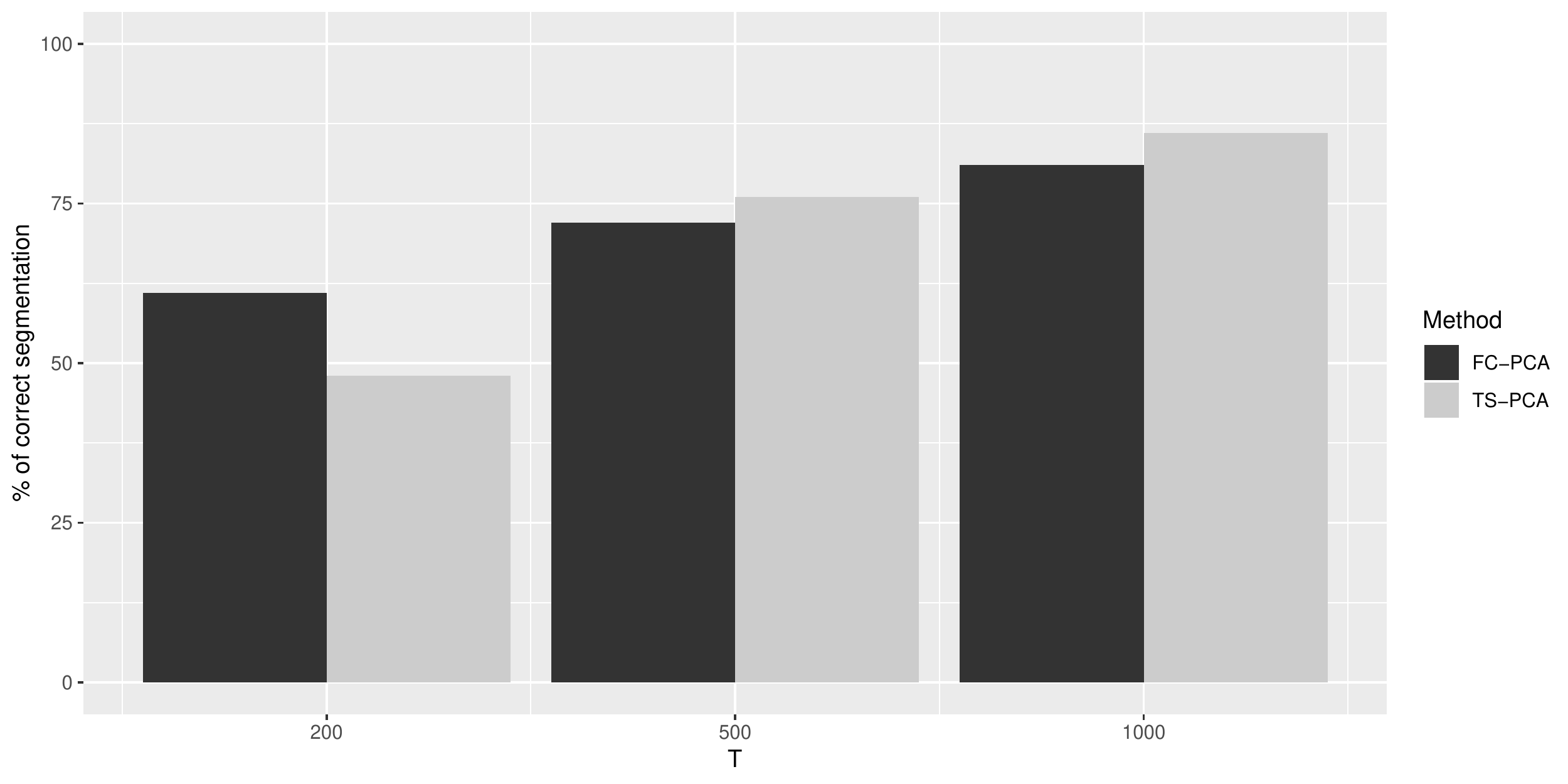}
\end{minipage}
\caption{\textbf{Model 1}, Left - Table of maximum $M^2$ and average $M^2$ errors defined in \eqref{eq:sim_error_definitions}  for the two competing methods: FC-PCA (proposed method) vs TS-PCA. Right - Percentage of \textit{correct segmentation} defined above for the two competing methods.} \label{fig:table_figure_model1}
\end{figure}

\noindent \textbf{Model 2}: We take $p=6$, $m=3$ and the components of $Y_t$ are given by $Y_{k,t} = z_{1,t+k-1}$ for $k=1,2,3$, $Y_{k,t} = z_{2,t+k-4}$ for $k=4,5$ and $Y_{k,t} = z_{3,t}$ for $k=6$. Here, $z_{1,t}$ follows a ARMA$(1,2)$ with AR coefficients $(0.9)$ and MA coefficients $(0.8,-0.2)$, $z_{2,t}$ follows a AR$(3)$ with AR coefficients $(1.25,-0.75,0.3)$ and $z_{3,t}$ follows a MA$(3)$ with MA coefficients $(1,-1,-0.8)$. The innovation terms in all the ARMA processes above are assumed to follow i.i.d $N(0,1)$.

\begin{figure}[H]
\begin{minipage}{0.45\textwidth}
\begin{tabular}{|c|c|c|c|}
\hline
T & Method & max. $M^2$ & avg. $M^2$ \\
\hline
\multirow{2}{*}{200} & FC-PCA & 0.077 & 0.046 \\
 & TS-PCA & 0.425  & 0.213   \\
\hline
\multirow{2}{*}{500} & FC-PCA & 0.023 & 0.013  \\
 & TS-PCA & 0.324 & 0.148   \\
\hline
\multirow{2}{*}{1000} & FC-PCA & 0.009 & 0.006  \\
 & TS-PCA & 0.216 & 0.104  \\
\hline
\end{tabular}
\end{minipage}
\begin{minipage}{0.5\textwidth}
\includegraphics[scale=0.4]{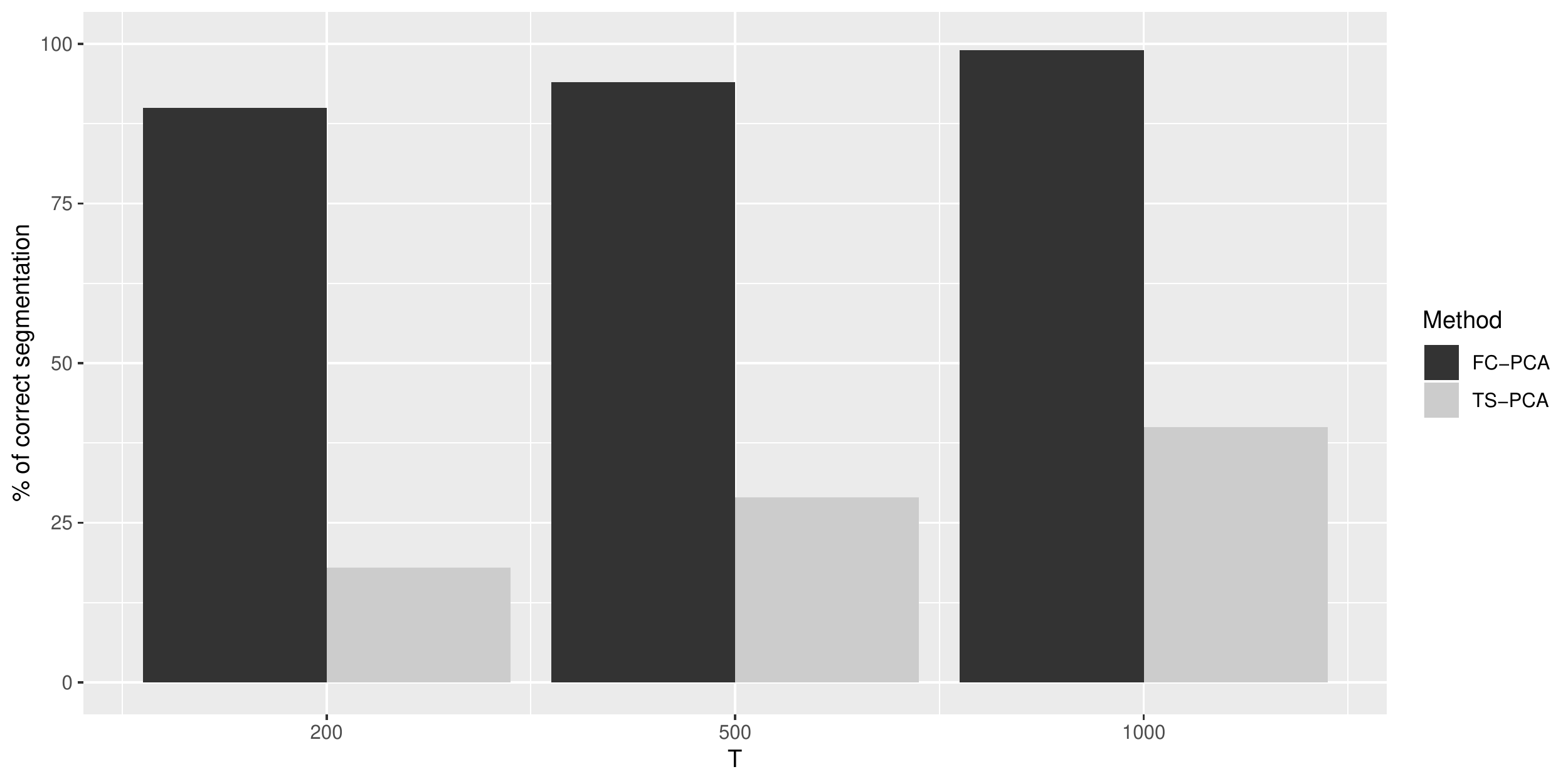}
\end{minipage}
\caption{\textbf{Model 2}, Left - Table of maximum $M^2$ and average $M^2$ errors defined in \eqref{eq:sim_error_definitions}  for the two competing methods: FC-PCA (proposed method) vs TS-PCA. Right - Percentage of \textit{correct segmentation} defined above for the two competing methods.} \label{fig:table_figure_model2}
\end{figure}

\noindent \textbf{Model 3}: We take $p=9$, $m=3$, $\phi = (1,0.7,-0.5,0.2)$, $\eta = (1,-0.9)$ and the components of $Y_t$ are given by $Y_{k,t} = \phi_k z_{1,t+k-1}$ for $k=1,2,3,4$, $Y_{k,t} = z_{2,t+k-5}$ for $k=5,6,7$ and $Y_{k,t} = \eta_k z_{3,t+k-8}$ for $k=8,9$. Here, $z_{1,t}$ follows a AR$(1)$ with AR coefficients $(0.45)$ and innovations following i.i.d $N(0,3)$, $z_{2,t}$ follows a ARMA$(2,3)$ with AR coefficients $(0.8,-0.5)$ and MA coefficients $(1,0.8,1.8)$ and innovations following i.i.d $N(0,5)$,  $z_{3,t}$ follows a ARMA$(2,2)$ with AR coefficients $(-0.7,-0.5)$ and MA coefficients $(-1,-0.8)$ and innovations following i.i.d $N(0,1)$.

\begin{figure}[H]
\begin{minipage}{0.45\textwidth}
\begin{tabular}{|c|c|c|c|}
\hline
T & Method & max. $M^2$ & avg. $M^2$ \\
\hline
\multirow{2}{*}{200} & FC-PCA & 0.021 & 0.015 \\
 & TS-PCA & 0.823  & 0.681  \\
\hline
\multirow{2}{*}{500} & FC-PCA & 0.013 & 0.009  \\
 & TS-PCA & 0.762 & 0.668   \\
\hline
\multirow{2}{*}{1000} & FC-PCA & 0.006 & 0.004  \\
 & TS-PCA & 0.727 & 0.654  \\
\hline
\end{tabular}
\end{minipage}
\begin{minipage}{0.5\textwidth}
\includegraphics[scale=0.4]{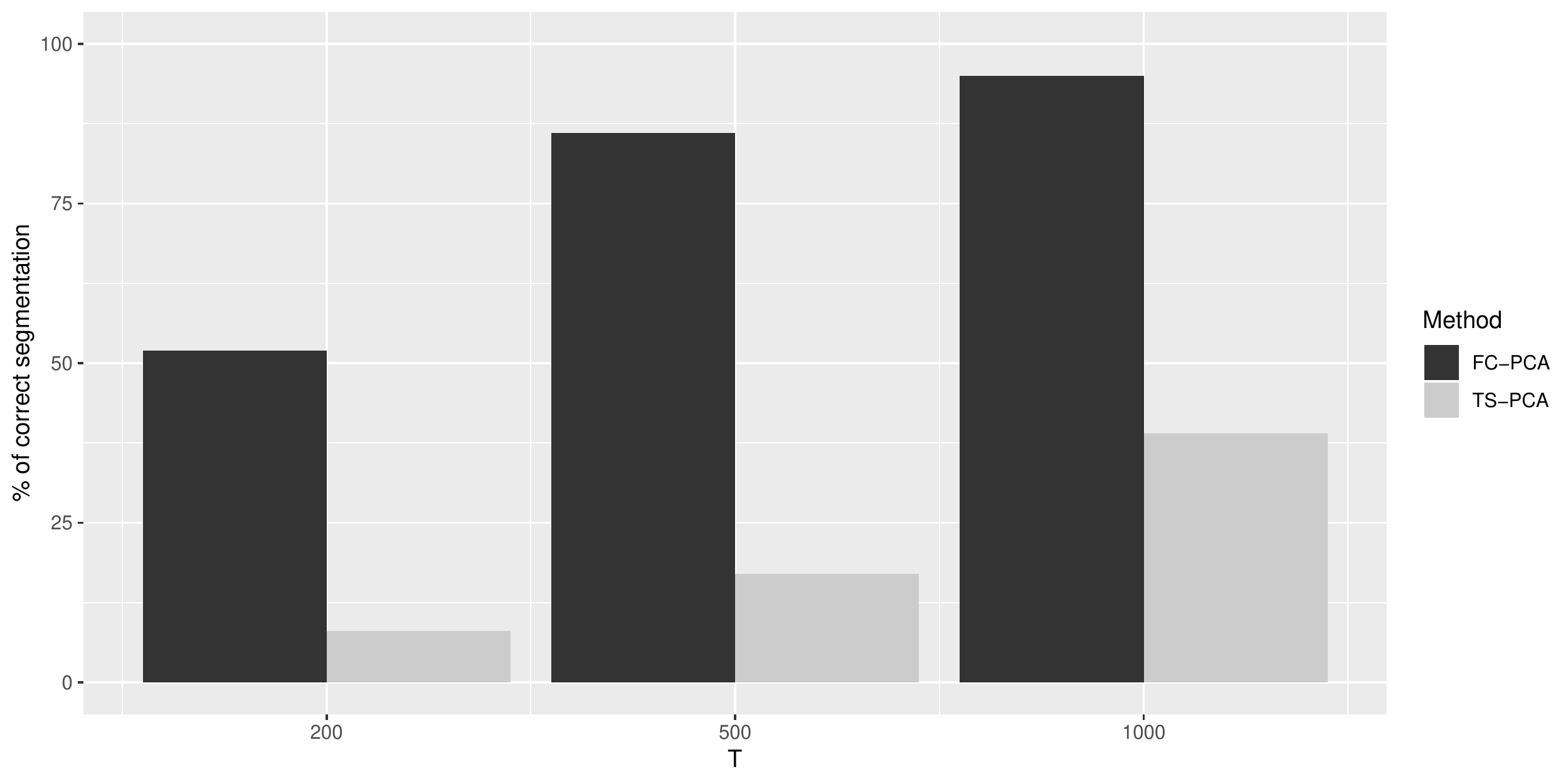}
\end{minipage}
\caption{\textbf{Model 3}, Left - Table of maximum $M^2$ and average $M^2$ errors defined in \eqref{eq:sim_error_definitions}  for the two competing methods: FC-PCA (proposed method) vs TS-PCA. Right - Percentage of \textit{correct segmentation} defined above for the two competing methods.} \label{fig:table_figure_model3}
\end{figure}

\noindent \textbf{Model 4}: We take $p=9$, $m=3$ and the components of $Y_t$ are given by $Y_{k,t} = z_{1,t+k-1}$ for $k=1,2,3,4$, $Y_{k,t} = z_{2,t+k-5}$ for $k=5,6,7$ and $Y_{k,t} =  z_{3,t+k-8}$ for $k=8,9$. Here, $z_{1,t}$ follows a ARMA$(2,4)$ with AR coefficients $(-0.4,0.5)$ and MA coefficients $(1,0.8,1.5,1.8)$, $z_{2,t}$ follows a ARMA$(2,3)$ with AR coefficients $(0.85,-0.3)$ and MA coefficients $(1,0.5,1.2)$ and $z_{3,t}$ follows a ARMA$(2,1)$ with AR coefficients $(0.9,-0.6)$ and MA coefficient $(0.5)$. The innovation terms in all the ARMA processes above are assumed to follow i.i.d $N(0,1)$. 

\noindent \textbf{Model 5}: We take $p=7$, $m=1$ and the components of $Y_t$ are given by $Y_{k,t} = z_{1,t+k-1}$ for $k=1,2,\hdots,7$. Here, $z_{1,t}$ follows a ARMA$(1,3)$ with AR coefficients $(0.75)$ and MA coefficients $(1,-0.7,-0.6)$. The innovation term is assumed to follow i.i.d $N(0,1)$.

\begin{figure}[H]
\begin{minipage}{0.45\textwidth}
\begin{tabular}{|c|c|c|c|}
\hline
T & Method & max. $M^2$ & avg. $M^2$ \\
\hline
\multirow{2}{*}{500} & FC-PCA & 0.062 & 0.045 \\
 & TS-PCA & 0.213  & 0.152  \\
\hline
\multirow{2}{*}{1000} & FC-PCA & 0.023 & 0.016  \\
 & TS-PCA & 0.151 & 0.102   \\
\hline
\multirow{2}{*}{2000} & FC-PCA & 0.011 & 0.008  \\
 & TS-PCA & 0.012 & 0.083 \\
\hline
\end{tabular}
\end{minipage}
\begin{minipage}{0.5\textwidth}
\includegraphics[scale=0.4]{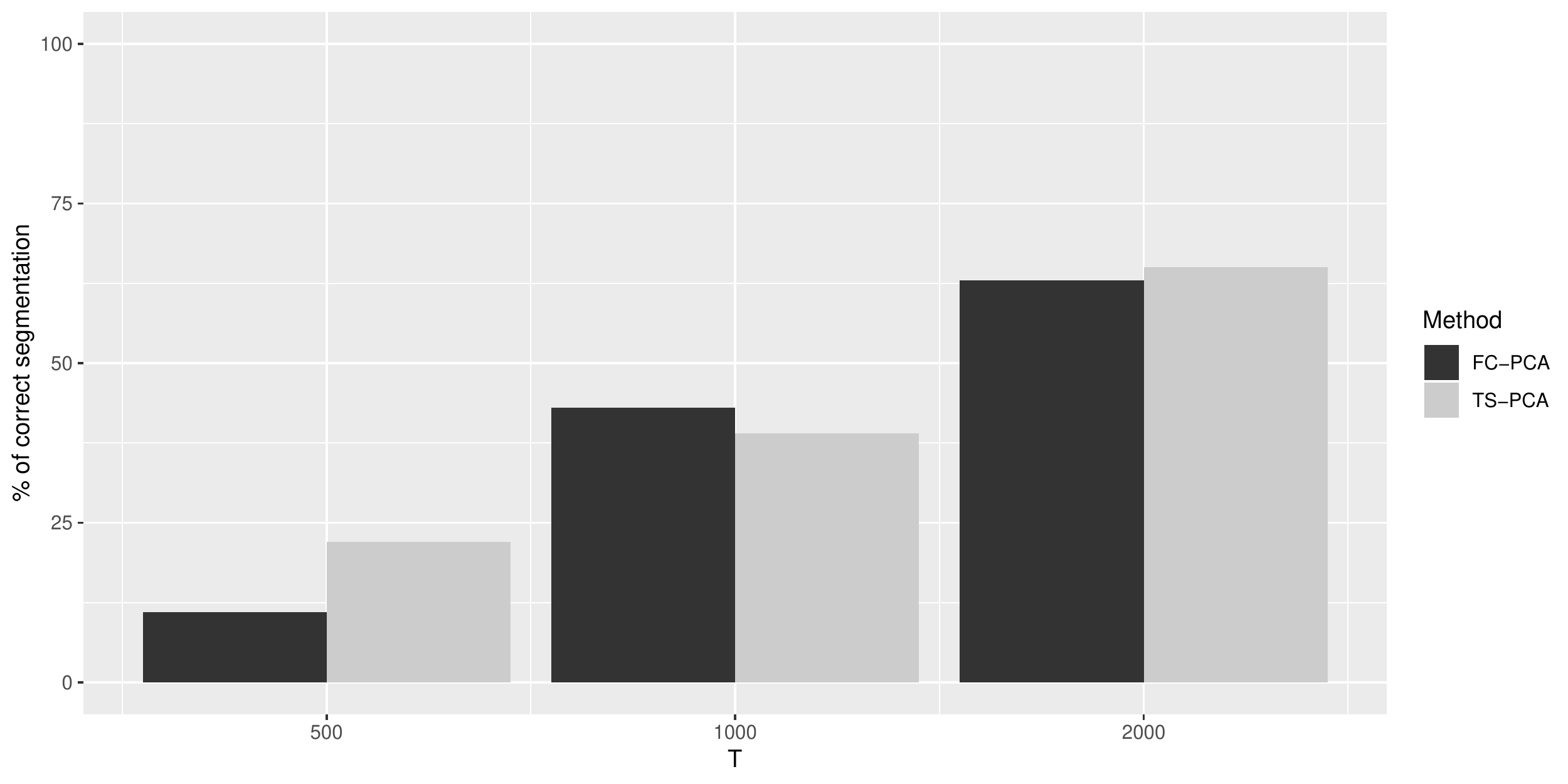}
\end{minipage}
\caption{\textbf{Model 4}, Left - Table of maximum $M^2$ and average $M^2$ errors defined in \eqref{eq:sim_error_definitions}  for the two competing methods: FC-PCA (proposed method) vs TS-PCA. Right - Percentage of \textit{correct segmentation} defined above for the two competing methods.} \label{fig:table_figure_model4}
\end{figure}

\begin{figure}[H]
\centering
\includegraphics[scale=0.42]{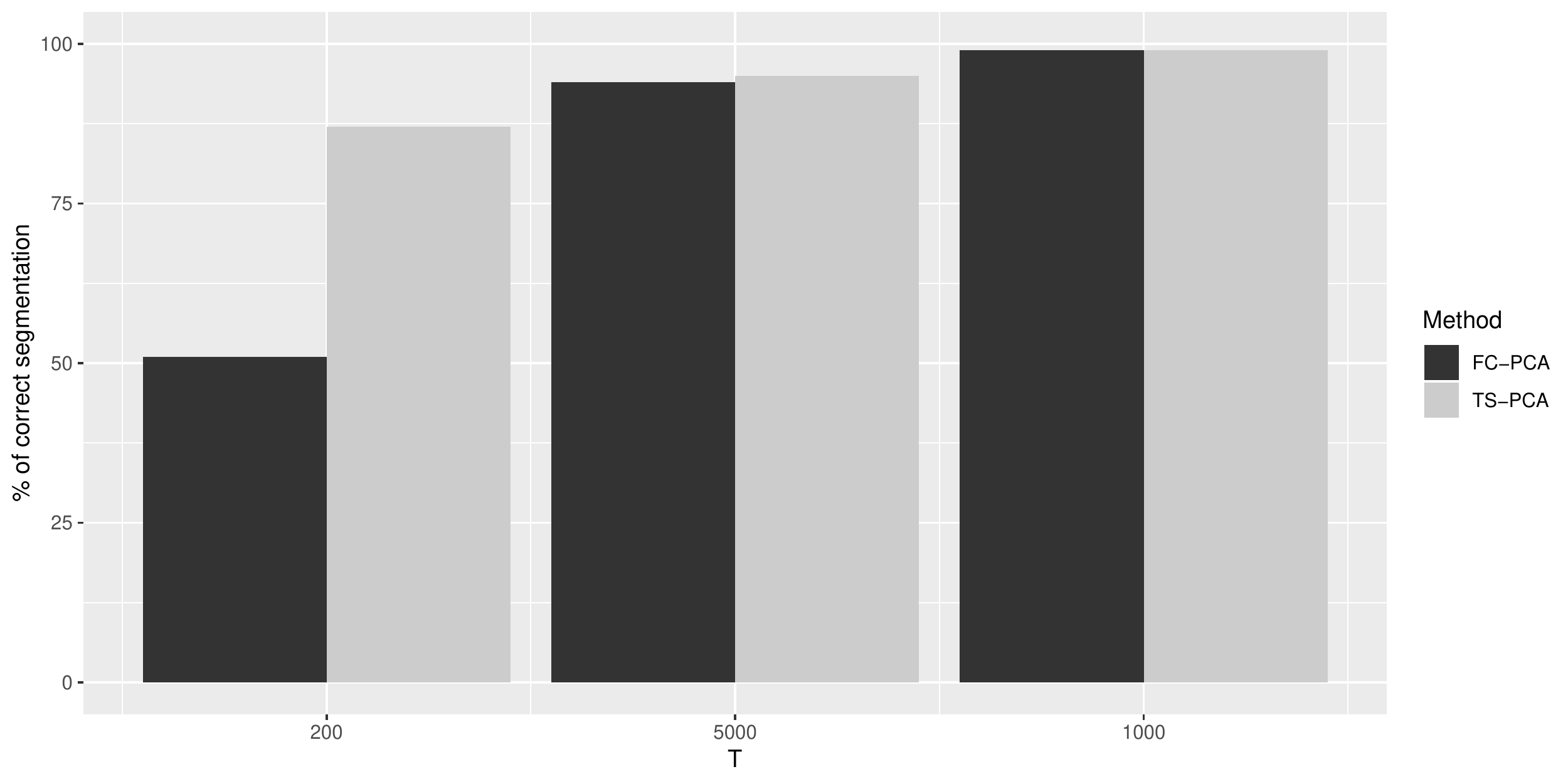}
\caption{\textbf{Model 5}  Percentage of \textit{correct segmentation} defined above for the two competing methods.} \label{fig:table_figure_model5}
\end{figure}

Next, in the left tables in Figures \ref{fig:table_figure_model1}-\ref{fig:table_figure_model4} we present the average and maximum errors due to the measure $M^2$ defined in \eqref{eq:sim_error_definitions}. Results are provided for the proposed frequency components PCA, denoted as FC-PCA, and the competing method from  \citet{yao2018} denoted as TS-PCA. These error measures are averaged over repetitions where the competing methods output the \textit{correct segmentation} for the various models. We first observe that in all models as $T$ increases, the error magnitudes decrease for both methods. We then see that at series lengths $T=500,1000$, FC-PCA performs better than TS-PCA in all the models in both error measures. In the right plots in Figures \ref{fig:table_figure_model1}-\ref{fig:table_figure_model4} we plot the percentage of \textit{correct segmentation} among 200 replications of the various models. We notice a comparable performance of the two methods for Model 1, better performance of FC-PCA for Models 2 and 3. In Model 4, TS-PCA has a higher percentage than FC-PCA for smaller sample sizes but the FC-PCA has better error measures reported in the table in Figure \ref{fig:table_figure_model4}. In Models 2 and 3, the third component $z_{3,t}$ has a strong pseudo-periodocity and it is particularly in such models that time domain methods such as TS-PCA show inferior performance.  For example, if one simulates $T=500$ observations from a univariate MA(3) model with coefficients $(1,-1,-0.8)$ (same as $z_{3,t}$ of Model 2), a Yule-Walker AR model fit yields a fitted AR$(d)$ with order $d \geq 10$. The pre-whitening suggested in \citet{yao2018} where one fits an AR model and obtains residuals becomes very challenging and affects finite sample performance.

In relation to Figure \ref{fig:eigengap_models} and Assumption \ref{as:eigengap} in Section \ref{sec:theory}, we witness that in general, FC-PCA performs better in models with a larger eigengap. We see from Figure \ref{fig:eigengap_models} that Model 4 has the smallest eigengap among the block matrices of $S_Y$ and the results in Figure \ref{fig:table_figure_model4} reflects the struggle by both methods in the percentage of \textit{correct segmentation}. In Model 5, we consider the case $m=1$ (no lower dimensional subseries). Figure \ref{fig:table_figure_model5} plots the percentage of of \textit{correct segmentation} among 200 replications of this model and we notice at series lengths $T=500,1000$, a comparable performance between the two competing methods.

\section{Application to wind speed forecasting}
\label{sec:application}

In this section we illustrate an application of our method in modeling and forecasting wind speed data. Wind power has become an important source of renewable energy with its obvious advantages such as being environmental friendly. Modeling and forecasting wind data is critical in successfully utilizing this energy source. In particular, forecasting enables operators with the ability to better plan for breaks in the influx of energy via this source. In order to maintain a reliable energy network, wind data gathered hourly and sub-hourly are both useful and necessary for conventional and renewable generators in the energy network; see the discussion in \citet{wind_speed_harsha}. Time series analysis has often been used in engineering to model both hourly and sub-hourly wind speed data. Hourly wind time series data has, in the past, been  modeled using autoregressive (AR) processes. In order to model wind turbine generators, \citet{wind_speed_ar} fit an AR time series model  to the hourly wind speed data and then simulate from this fitted model to understand the nonlinear relationship between power output and wind speeds. \citet{wind_speed_var} utilize a vector autoregressive time series model (VAR) for modeling hourly wind speeds recorded at 20 different zones in the UK. The fitted VAR model was then used for a sequential Monte Carlo simulation that led to an adequacy analysis which tests the reliability of the energy system.   \citet{wind_speed_harsha} fit a VAR model to sub-hourly wind speed data recorded at wind farm  locations in a given geographical region. In their work the sub-hourly data was seen to be nonstationary over longer time periods. To overcome this difficulty, VAR models were fit to smaller time segments wherein a stationarity assumption is more reasonable. The fitted VAR models across these different time segments are then used for further simulations from these models and helps with an economic dispatch formulation that assists in planning and management. \citet{dowell_wind} is another work that fits stationary time series  models to wind speed data over small time windows and comparisons are made with a trend plus VAR component model. 

We gather wind speed data recorded every 5 minutes (sub-hourly) from $p=7$ locations in the south eastern part of Oklahoma\footnote{Data source: National Renewable Energy Laboratory - \texttt{https://www.nrel.gov/grid/eastern-wind-data.html}} during 2011. Similar to the approach in \citet{wind_speed_harsha}, we model the wind speed data over time segments of 13 hours using a trend plus stationary component model. More precisely, let the vector $W_t$, $t=1,2,\hdots,T=156$, be the wind speed (in metres per second $m/s$) observed at any given time $t$ recorded at 100m above sea level. Here, the series length $T=156$ corresponds to 13 hours worth of sub-hourly wind speed data. Seasonality is another artifact that appears often in wind data due to changing wind patterns over days and longer periods. But because we consider sub-hourly data over a short time segment (13 hours), we do not focus on seasonal changes. We thus have $W_t = G_t + X_t$ where the vector $G_t$ is viewed as the location specific trend component and $X_t$ is the stationary component. \citet{hill_var_wind} also consider a location specific trend plus stationary component model and use a VAR model fit for the stationary part. The above model is also similar to the intercept plus autoregressive component model considered in \citet{lenzi_wind} for wind power data. In forecasting, we consider the first 13 hours of each day during September 2011, forecast the next two time steps and compare our performance with the conventional stationary VAR model, component-wise univariate AR model and the TS-PCA method from \citet{yao2018}. 

As an illustration, we consider wind speed data from the first 13 hours of September 5, 2011. First, we plot the autocorrelation function (ACF) of the 7 components of $X_t$ in Figure \ref{fig:acf_orig}. We observe a lot of cross-covariances at various lags between the 7 components. Next, we apply our FC-PCA and obtain the segmented series $Y_t$. The estimated graph $G$ from our FC-PCA method identifies 5 groups namely (1,2), (3,4), (5), (6) and (7). Figure \ref{fig:acf_seg} plots the ACF of the transformed (segmented) series $Y_t$. We notice a clear block structure from this plot wherein the first two components have cross-covariances, similarly the next two components and finally the remaining components, that were grouped as single entities by our method, do not exhibit significant cross-covariance with the other components. Next, we fitted univariate AR models to  the 7 components of $X_t$ and the resulting model orders were $(8,5,3,3,10,2,5)$. The presence of high model orders makes pre-whitening a challenging task and makes a spectral domain approach more appropriate for such types of data. 

\begin{figure}[t]
\centering
\includegraphics[scale=0.7]{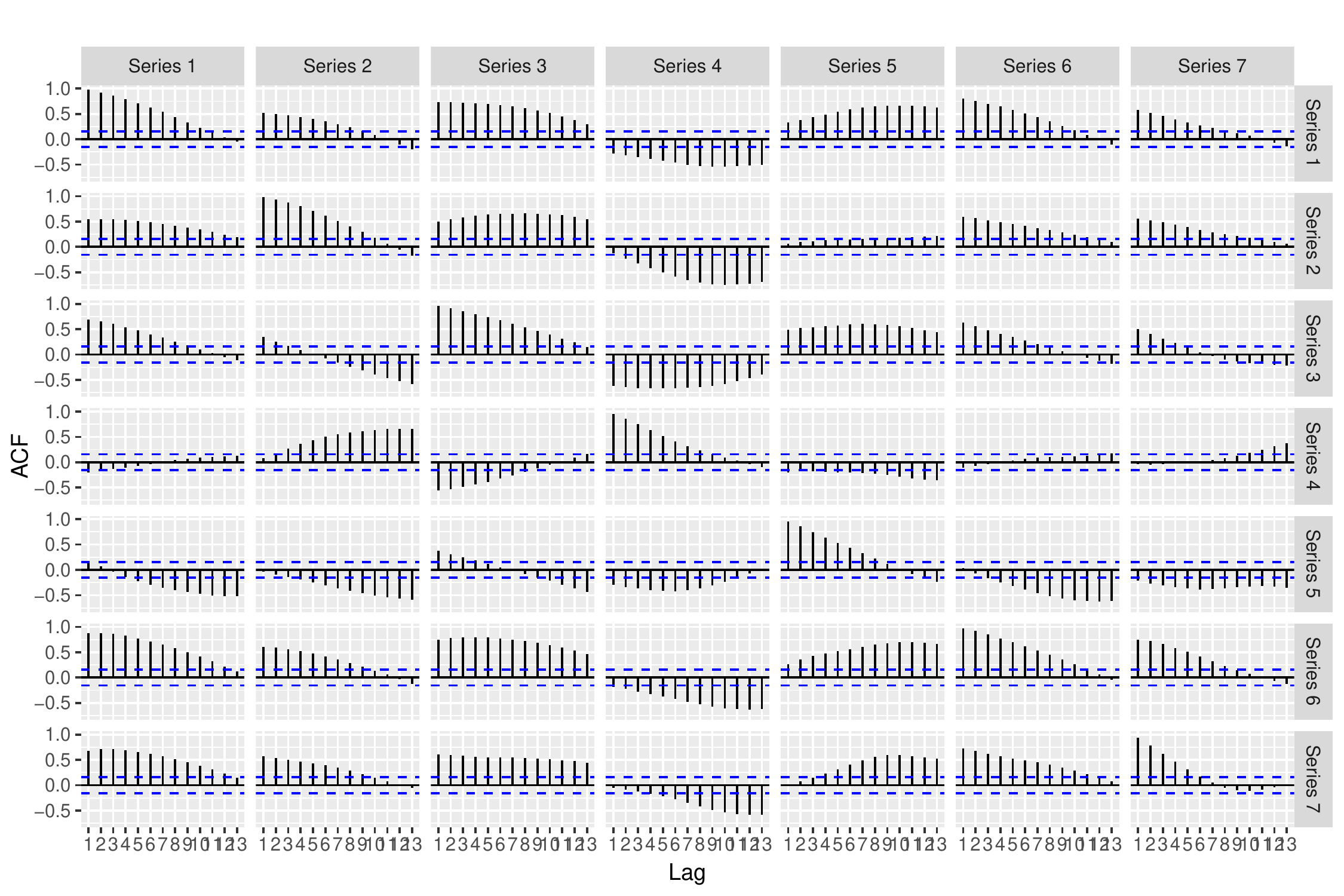}
\caption{ACF plot of 7 components of $X_t$. Data was from wind speed recorded  every 5 minutes between 12am and 11am on September 5, 2011.} \label{fig:acf_orig}
\end{figure}

\begin{figure}[t]
\centering
\includegraphics[scale=0.7]{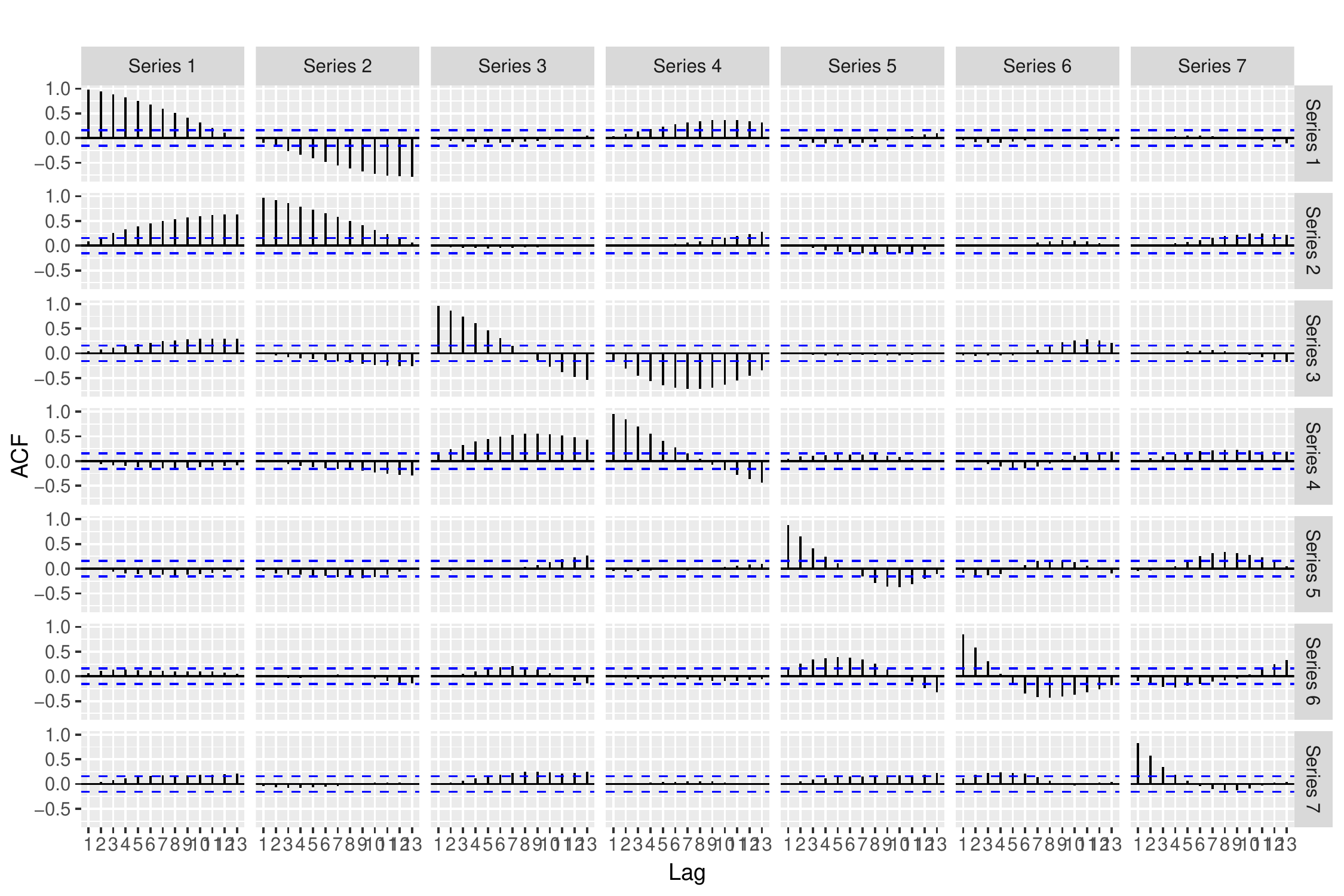}
\caption{ACF plot of 7 components of the transformed (segmented) series $Y_t$ after applying FC-PCA. Our method found the following 5 groups: (1,2), (3,4), (5), (6), (7). } \label{fig:acf_seg}
\end{figure}

Next, we compare the two steps ahead forecasting accuracy using mean squared error and standard deviation of the forecasts. Once we obtain the segmented series $Y_t$, we fit VAR models to the individual lower-dimensional subseries, compute the forecasts and then obtain forecasts of the original data using the estimated demixing matrix. We consider the first 13 hours of every single day in September 2011 and forecast wind speeds at hours 14 and 15 of those days. The MSE is obtained by computing the average of the squared error of this forecast over the 30 days of that month.  We compare forecasting performance our method (FC-PCA) with the stationary vector autoregressive model (VAR), the TS-PCA of \citet{yao2018} and univariate AR models fit to the individual components. For the FC-PCA method, we take the Bartlett Priestley kernel with bandwidth $h=T^{-0.1}$ for estimating the spectral matrix in \eqref{eq:kernel_spectral_estimator}. The level of significance is set to 0.01 in the pairwise testing from Section \ref{sec:permuation_subseries}. From Table \ref{tab:wind_foreast_accuracy} we observe that FC-PCA performs better than all the other methods in the one step and two steps ahead forecasts. The TS-PCA method was seen to be sensitive to the choice of the number of cross-covariance lags used in their method and we present only the best result there.

\begin{table}
\centering
\begin{tabular}{|c|c|c|c|}
\hline
 Method & MSE - One step ahead & MSE - Two steps ahead \\
\hline
  FC-PCA & 0.550 (0.007)   &  0.958 (0.042) \\
  VAR   &  0.610 (0.007)    &  1.276 (0.034) \\ 
  TS-PCA &  0.628 (0.008)   &  1.523 (0.067) \\
  Univariate AR & 1.678 (0.152) & 1.736 (0.154) \\
\hline
\end{tabular}
\caption{Mean squared error (MSE) for one step and two steps ahead forecasts for September 2011. The standard deviations are reported in the brackets.  } \label{tab:wind_foreast_accuracy}
\end{table}

\noindent Implementation Details: All computations in this paper were carried out using R-version 3.5.2 run on a Linux platform with a Intel(R) Xeon(R) CPU E5-2690 64-bit processor. As an example of running times we consider the application in Section \ref{sec:application} that carries out the different PCA methods for a $p=7$ dimensional series with a series length of $T=156$ (13 hours of sub-hourly wind speed data). On average, the FC-PCA takes 0.161 seconds, the TS-PCA takes 0.087 seconds and the traditional VAR takes 0.022 seconds.

\section{Concluding remarks}
\label{sec:conclusion}
In this work we proposed a new spectral domain method that finds contemporaneous linear transforms of an observed $p$-variate stationary time series. The transformed series consists of several lower-dimensional subseries such that components within a subseries have non-zero spectral coherence but components across different subseries have no spectral coherence. A two step procedure is described wherein first the eigenvectors of the sum of real parts of spectral matrices at different frequencies provide an initial solution to the transformed series. Second, a consistent test of zero spectral coherence on pairs of components is utilized to permute the components of the transformed series resulting in the desired segmentation. In addition to theoretical justifications, simulation studies are included to support the proposed method. In comparison to its time domain counterparts, our spectral domain approach avoids any pre-whitening, is hence completely nonparametric and is more appropriate for handling time series data with strong periodicities. We also show an application of our method in modeling and forecasting wind speed data and by comparing with other techniques we witness better forecasting performance of our method. 

The proposed method is fairly general can potentially be applied to numerous other real data applications such as modeling and predicting other renewable energy sources, forecasting disease progression and  predicting air pollution related variables such as $CO$ and $NO_2$. Extending the methodology to cover direct dependence (partial spectral coherence) is another interesting task to pursue. The large sample results in Section \ref{sec:theory} are derived for the fixed $p$ case. One future direction of interest is to devise a computationally feasible method as $p$ grows along with relevant theoretical justifications. Another direction of interest is in extending the methodology in Section \ref{sec:methodology} for locally stationary time series (\citealp{dahlhaus97}) wherein the asymptotic de-correlation property of frequency components (and spectral densities) at unequal frequencies no longer holds. 

\bibliographystyle{Chicago}
\bibliography{pca_timeseries}

\section*{Appendix: Proofs} \label{sec:proofs}

The proof of Theorem \ref{thm:colspace_consistency} follows from a application of Theorem 8.1.10 and Corollary 8.1.11 of \citet{golub2012} that are on the perturbation of invariant subspaces. We first state those two results in the following lemma. We denote $||R||_2$ as the $L_2$ norm of the matrix $R$.

\begin{lemma} \label{lemma:golub}
Suppose $A$ and $A+E$ are $p \times p$ symmetric matrices and that $Q=(Q_1,Q_2)$, where $Q_1$ is $p \times r$ and $Q_2$ is $p \times (p-r)$, is an orthogonal matrix such that $A \times \mathcal{C}(Q_1) \subset \mathcal{C}(Q_1)$ (invariant subspace for A). Partition the matrices $Q^{'} A Q$ and $Q^{'} E Q$ as 
$$ Q^{'} A Q = \begin{bmatrix}
D_1 & 0 \\
0 & D_1
\end{bmatrix}, \;\; Q^{'} E Q =  \begin{bmatrix}
E_{11} & E_{21}^{'} \\
E_{21} & E_{22}
\end{bmatrix} $$
If $sep(D_1,D_2) =  \min_{ \mu \in l( D_1 ), \; \nu \in l( D_2 ) }  |\mu - \nu| \; > \; 0 $ where $l( D_1 )$ and $l( D_2 )$ denote the set of eigenvalues of the matrices $D_1$ and $D_2$ respectively and $||E||_2 \leq sep(D_1,D_2)/5$, then 

\begin{itemize}
\item[(a).] There exists a $ (p-r) \times r$ matrix $P$ with $|| P ||_2 \leq 4 || E_{21} ||_2/sep(D_1,D_2) $ such that the columns of $\widehat{Q}_1 = (Q_1 + Q_2 P) (I + P^{'} P)^{-1/2}$ define an orthonormal basis for a subspace that is invariant for $A+E$.

\item[(b).]  $ || Q_1 -  \widehat{Q}_1 ||_2 \leq \frac{4}{sep(D_1,D_2)} || E_{21} ||_2$. 
\end{itemize}
\end{lemma}

\begin{proof}[\textbf{Proof of Theorem \ref{thm:colspace_consistency}}]
First, the distance measure $M$ defined in \eqref{eq:subspace_distance_metric} is such that for $A_i$, $i=1,2,\hdots,m$
\begin{gather} \label{eq:golub_consequence}
M \Big( \mathcal{C}(A_i) , \mathcal{C}(\widehat{A}_i)  \Big) = \sqrt{1 - \frac{1}{p}tr \Big( A_i A_i^{'} \widehat{A}_i \widehat{A}_i^{'} \Big) } 
\leq \sqrt{ || I_p -  A_i A_i^{'} \widehat{A}_i \widehat{A}_i^{'} ||_2 } \\ \notag
= \sqrt{ || A_i^{'} (A_i - \widehat{A}_i)(A_i - \widehat{A}_i)^{'} A_i ||_2  } \leq \sqrt{2 || A_i - \widehat{A}_i  ||_2^2 }
\end{gather}
Next, under Assumptions \ref{as:xt}, \ref{as:kernel} and Theorems 7.4.1-7.4.3 of \citet{brillinger81}, the kernel spectral estimator is such that for any $a,b=1,2,\hdots,p$
\begin{equation}
E || \widehat{f}(\omega) - f(\omega)  ||_2^{2r} = O \Big( \frac{1}{ T^r h^r} \Big) \;\; \textrm{and} \;\; E || \widehat{f}_{ab}(\omega) - f_{ab}(\omega)  ||_2 = O \Big( \frac{1}{ \sqrt{T h} } \Big)   
\end{equation}
uniformly in $\omega \in \lbrack - \pi , \pi \rbrack$. This implies that 
\begin{equation}
|| S_X - \widehat{S}_X  ||_2  = O_p \Big( \frac{1}{ \sqrt{Th} } \Big).
\end{equation}
Finally, an application of Lemma \ref{lemma:golub} yields 
\begin{equation}
\max_{i=1,2,\hdots,m} M \Big( \mathcal{C}(A_i) , \mathcal{C}( \widehat{A}_i )  \Big) = O_p \Big( || S_X - \widehat{S}_X ||_2 \Big) = O_p \Big( \frac{1}{ \sqrt{Th} } \Big).
\end{equation}
\end{proof}

\begin{proof}[\textbf{Proof of Theorem \ref{thm:adjacency_consistency}}]
With $E = (e_{a,b})$ as the true $p \times p$ adjacency matrix and its estimated version $\widehat{E} = (\widehat{e}_{a,b}) $, $a,b=1,2,\hdots,p$ we have, 

\begin{equation}
P\Big( \widehat{E} = E \Big) = P \Big( \bigcap_{a < b} \{ \widehat{e}_{a,b} = e_{a,b} \}  \Big)
\end{equation}

\noindent For every $a < b$ and $a,b=1,2,\hdots,p$, we have,
\begin{gather*} \label{eq:null_edge_asymptotic_result}
P( \widehat{e}_{a,b} > e_{a,b} ) = P \Big( \widehat{D}(\tilde{Y},a,b) > c(\alpha_{a,b}) \Big) \; \overset{T \rightarrow \infty}{\longrightarrow} \alpha_{a,b}, 
\end{gather*}
where $\alpha_{a,b}$ is the chosen level and $c(\alpha_{a,b})$ is the critical value from the null distribution given in \eqref{eq:test_statistic_null_distribution}. Similarly, the consistency result from \eqref{eq:test_statistic_alt_result} implies that for some positive constant $C>0$, 
\begin{gather*} \label{eq:alt_edge_asymptotic_result}
P( \widehat{e}_{a,b} < e_{a,b} ) = 1 - P \Big( \widehat{D}(\tilde{Y},a,b) > C \Big) \overset{T \rightarrow \infty}{\longrightarrow} 0.  
\end{gather*}

Now, for part (a) when $m=1$, 

\begin{gather*}
P \Big( \bigcap_{a < b} \{ \widehat{e}_{a,b} = e_{a,b} \}  \Big) = 1 - P \Big( \bigcup_{a < b} \{ \widehat{e}_{a,b} < e_{a,b} \}  \Big) \geq 1 - \sum_{a<b} P \Big( \{ \widehat{e}_{a,b} < e_{a,b} \}  \Big) \\
= 1 - o(1),
\end{gather*}
for a fixed $p$. For part (b) when $m>1$,
\begin{gather*}
P \Big( \bigcap_{a < b} \{ \widehat{e}_{a,b} = e_{a,b} \}  \Big) = 1 - P \Big( \bigcup_{a < b} \{ \widehat{e}_{a,b} \neq e_{a,b} \}  \Big) 
\geq 1 - \sum_{a<b} P \Big( \{ \widehat{e}_{a,b} > e_{a,b} \}  \Big) - o(1) \\
 =   1 - \sum_{a<b} \alpha_{a,b} - o(1). 
\end{gather*} 
 
\end{proof}

\end{document}